\documentclass[submission,copyright,creativecommons]{eptcs}
\providecommand{\event}{QPL 2022}

\usepackage{breakurl}
\usepackage{hyperref}
\usepackage{amsmath,amssymb,amsthm,bbold}
\usepackage{stmaryrd}
\usepackage{mathrsfs}
\usepackage{soul}
\usepackage{color}
\usepackage{enumerate}
\usepackage{tikz-cd}
\usepackage{tikzit}
\usepackage{physics}
\usepackage{quiver}

\newcommand{\dffack}{Supported by \emph{DFF $\mid$ Natural Sciences} International Postdoctoral Fellowship 0131-00025B}
\newcommand{\epsrcack}{Supported by EPSRC Fellowship EP/R044759/1}
\newcommand{\cdtack}{Supported by EPSRC grant EP/L01503X/1 via the CDT in Pervasive Parallelism}

\newcommand{\id}{\mathrm{id}}
\newcommand{\opp}{\mathrm{op}}

\newcommand{\tot}{\xrightarrow}

\newcommand{\im}{\mathrm{im}}
\newcommand{\medeq}[1]{\sim_{{#1}_\oplus}}

\newcommand{\cD}{\mathcal{D}}
\newcommand{\cE}{\mathcal{E}}
\newcommand{\cF}{\mathcal{F}}

\newcommand{\cat}[1]{\ensuremath{\mathbf{#1}}}
\newcommand{\C}{\cat{C}}
\newcommand{\D}{\cat{D}}
\newcommand{\DagMon}{\mathrm{DagMon}}
\newcommand{\FinBij}{\cat{FinBij}}
\newcommand{\FinPInj}{\cat{FinPInj}}
\newcommand{\Unitary}{\cat{Unitary}}
\newcommand{\Isometry}{\cat{Isometry}}
\newcommand{\coIsometry}{\cat{coIsometry}}
\newcommand{\Contraction}{\cat{Contraction}}

\newcommand{\FHilbCPTP}{\cat{FHilb}_{\mathrm{CPTP}}}
\newcommand{\FHilbCPTN}{\cat{FHilb}_{\mathrm{CPTN}}}

\newcommand{\FCstarCPTN}{\cat{FCstar}_{\mathrm{CPTN}}}

\newcommand{\La}{\ensuremath{L_\oplus}}
\newcommand{\Ra}{\ensuremath{R_\oplus}}
\newcommand{\LRa}{\ensuremath{{LR}_\oplus}}
\newcommand{\Lm}{\ensuremath{L_\otimes}}

\newcommand{\Lt}{\ensuremath{\Lm^t}}

\newcommand{\SplitM}{\ensuremath{\mathbf{Split}^{\discard}}}

\usepackage{tikz}
\usetikzlibrary{circuits.ee.IEC}

\newcommand\discard{%
\mathbin{\text{\begin{tikzpicture}[circuit ee IEC,yscale=0.9,xscale=0.6,rotate=180]
\draw (0,2ex) to (0,0) node[ground,rotate=-90,xshift=.65ex] {};
\end{tikzpicture}}}%
}

\newcommand{\B}{\mathcal{B}}

\newtheorem{theorem}{Theorem}
\newtheorem{proposition}[theorem]{Proposition}
\newtheorem{lemma}[theorem]{Lemma}
\newtheorem{corollary}[theorem]{Corollary}
\theoremstyle{definition}
\newtheorem{definition}[theorem]{Definition}

\makeatletter
\newtheorem*{rep@theorem}{\rep@title}
\newcommand{\newreptheorem}[2]{%
\newenvironment{rep#1}[1]{%
 \def\rep@title{#2 \ref{##1}}%
 \begin{rep@theorem}}%
 {\end{rep@theorem}}}
\makeatother
\newreptheorem{proposition}{Proposition}
\newreptheorem{theorem}{Theorem}

\usepackage{tikzit}

\tikzstyle{scalar}=[fill=white, draw=black, shape=circle, tikzit fill=white, tikzit draw=black]
\tikzstyle{map}=[fill=white, draw=black, shape=rectangle, tikzit fill=white, tikzit draw=black, minimum width=4em]
\tikzstyle{bullet}=[fill=black, draw=none, shape=circle, tikzit fill=black, inner sep=1pt]
\tikzstyle{map1}=[fill=white, draw=black, shape=rectangle, tikzit fill=white, tikzit draw=black, minimum width=2em]
\tikzstyle{state}=[fill=white, draw=black, tikzit fill=white, tikzit draw=black, regular polygon, regular polygon sides=3, shape border rotate=180]
\tikzstyle{map3}=[fill=white, draw=black, shape=rectangle, tikzit fill=white, tikzit draw=black, minimum width=6em]
\tikzstyle{circle}=[fill=white, draw=black, shape=circle, tikzit fill=white, tikzit draw=black, inner sep=1pt]
\tikzstyle{label}=[font={\tiny}]
\tikzstyle{effect}=[fill=white, draw=black, shape=circle, tikzit fill=white, tikzit draw=black, regular polygon, regular polygon sides=3]
\tikzstyle{clone}=[fill=black, draw=black, shape=circle, inner sep=2pt]
\tikzstyle{map4}=[fill=white, draw=black, shape=rectangle, tikzit fill=white, tikzit draw=black, minimum width=9em]

\tikzstyle{invis}=[-, dashed, dash pattern=on 0.5mm off 0.5mm]

\begin{document}

\title{Universal Properties of Partial Quantum Maps}
\author{Pablo Andrés-Martínez\footnote{\cdtack}
\institute{University of Edinburgh \\ Quantinuum}
\and
Chris Heunen\footnote{\epsrcack}
\institute{University of Edinburgh}
\and 
Robin Kaarsgaard\footnote{\dffack}
\institute{University of Southern Denmark}
}

\def\titlerunning{Universal Properties of Partial Quantum Maps}
\def\authorrunning{P. Andrés-Martínez, C. Heunen, \& R. Kaarsgaard}

\maketitle

\begin{abstract}
  We provide a universal construction of the category of finite-dimensional
  C*-algebras and completely positive trace-nonincreasing maps from the rig
  category of finite-dimensional Hilbert spaces and unitaries. This
  construction, which can be applied to any dagger rig category, is described 
  in three steps, each associated with their own universal property, and draws
  on results from dilation theory in finite dimension. In this way, we
  explicitly construct the category that captures hybrid quantum/classical
  computation with possible nontermination from the category of its reversible
  foundations. We discuss how this construction can be used in the design and
  semantics of quantum programming languages.
\end{abstract}

\section{Introduction}
The account of quantum measurement offered by \emph{decoherence} establishes that the irreversible nature of mixed-state evolution occurs when a system is considered in isolation from its environment. When the environment is brought back into view, mathematically through techniques such as \emph{quantum state purification} and \emph{Stinespring dilation}, the reversible underpinnings of mixed-state evolution are exposed.

This perspective has in recent years led to the study of quantum theory through categorical completions of its reversible foundations, the category of finite-dimensional Hilbert spaces and unitaries, demonstrating connections between \emph{universal} constructions and effectful quantum programming~\cite{heunenkaarsgaard:qie}. 
This article constructs in a universal way the category of finite-dimensional C*-algebras and \emph{partial quantum channels} (completely positive trace-nonincreasing maps) from the rig category of finite-dimensional Hilbert spaces and unitaries. The construction has three stages, each with a universal property of its own.
\begin{itemize}
  \item Freely allowing \emph{partiality respecting the dagger structure} 
  (by making the additive unit a \emph{zero object}) allows
  contractive maps to be described by unitaries through \emph{Halmos
  dilation}~\cite{halmos:normal,robinson:julia,levyshalit:dilationsfd}.
  \item Freely allowing the \emph{hiding of states} in a way that
  \emph{respects partiality} (by making the multiplicative unit \emph{terminal}
  for \emph{total} maps) allows completely positive trace-nonincreasing maps to
  be described through contractions, using a variant of \emph{Stinespring
  dilation}~\cite{stinespring:positive}. This construction has an interesting
  universal property as a pushout of monoidal categories.
  \item Freely splitting certain idempotents on finite-dimensional
  Hilbert spaces yields finite-dimensional C*-algebras, which describe \emph{hybrid quantum/classical} computation.
\end{itemize}
All three universal constructions are abstract and apply to any suitably structured category. 
They show that the traditional model of C*-algebras inevitably arises from the mere concepts of quantum circuits, partiality, hiding, and classical communication, without any concept of \textit{e.g.}\ norm. 
Thus they inform the design of quantum programming languages~\cite{heunenkaarsgaard:qie}, as part of a highly effective broader approach to program semantics from universal properties~\cite{wadler:free,hutton:tutorial,statonlevy:universal}.

\paragraph{Related work}
The role of universal properties and categorical completions in quantum theory,
in particular the \emph{monoidal indeterminates}
construction~\cite{hermidatennent:indeterminates}, has been studied in recent
years~\cite{rennelastatonfurber:infinite, huotstaton:universal,
huotstaton:completion, comfort:zxamp, heunenkaarsgaard:bennett,
heunenkaarsgaard:qie} (see also \cite{steinstaton:compositional} for a related
approach in the probabilistic case, and \cite{westerbaan:paschke,
parzygnat:stinespring} for other accounts of dilations as universal
properties). The role of partiality in effectus theory was studied in
\cite{cho:partial}. We deepen the connections between dilations and universal
properties of functors between categories of quantum systems first observed in
\cite{huotstaton:universal,huotstaton:completion}. A direct connection between
universal constructions and the design and semantics of quantum programming
languages with effects was demonstrated in \cite{heunenkaarsgaard:qie}. The
idempotent splitting of the category of finite-dimensional Hilbert spaces and
completely positive maps, in particular the fact that it contains all
finite-dimensional C*-algebras, has been the subject of study and discussion in
\cite{heunenkissingerselinger:cpproj,rennelastatonfurber:infinite,coeckeselbytull:tworoads}.

\paragraph{Overview}
Section~\ref{sec:affine_completion} recalls some facts about rig categories and their additive affine completion and relation to partiality, which we extend in Section~\ref{sec:biaffine_completion} to the biaffine completion and dagger partiality, and show that the biaffine completion of the category $\Unitary$ of (finite-dimensional Hilbert spaces and) unitaries is precisely $\Contraction$ of contractive maps. Section~\ref{sec:generalized_pablo_pushout} describes a construction that completes $\Contraction$ to the category $\FHilbCPTN$ of \emph{partial quantum channels} (completely positive trace-nonincreasing maps) using a variant of Stinespring dilation, and we show that this construction satisfies a universal property as a pushout in the category of monoidal categories. Finally, Section~\ref{sec:splitting_measurements} shows that splitting along a particular class of idempotents, corresponding in $\FHilbCPTN$ to \emph{measurements}, completes $\FHilbCPTN$ to the category $\FCstarCPTN$ of finite-dimensional C*-algebras and partial quantum channels. We end in Section~\ref{sec:discussion} with a discussion of the applications of these constructions to the design and semantics of quantum programming languages.

\section{The additive affine completion of rig categories}
\label{sec:affine_completion}

We enter the story assuming that the reader is familiar with monoidal categories and dagger categories~\cite{heunenvicary:cqt}, and proceed with preliminaries about rig categories and their additive affine completion.

A \emph{rig category} is a category which is symmetric monoidal in two different ways, such that one monoidal product distributes over the other, subject to a large amount of coherence equations~\cite{laplaza:coherence}. In analogy with the situation in Hilbert spaces, we usually write these monoidal products as $(\otimes, I)$ (the ``tensor product'') and $(\oplus, O)$ (the ``direct sum'') with $\otimes$ distributing (up to natural isomorphism) over $\oplus$ via distributors
$\delta^L \colon A \otimes (B \oplus C) \to (A \otimes B) \oplus (A \otimes C)$ and $
\delta^R \colon (A \oplus B) \otimes C \to (A \otimes C) \oplus (B \otimes C)$
and annihilators (``nullary distributors'')
$\delta_0^L \colon O \otimes A \to O$ and $\delta_0^R \colon A \otimes O \to O.$
A \emph{dagger rig category} is a dagger category with a rig structure such that all coherence isomorphisms are unitary (\textit{i.e.}, satisfy $f^{-1} = f^\dagger$). Natural examples of rig categories are \emph{distributive categories} where $\otimes$ is a categorical product with a terminal object as its monoidal unit and $\oplus$ is a categorical coproduct with an initial unit. However, not all rig categories are of this form: the category $\Unitary$ of finite-dimensional Hilbert spaces and unitaries (with tensor product and direct sum) and the category $\FinBij$ of finite sets and bijections (with cartesian product and disjoint union) are both (dagger) rig categories, but neither has products or coproducts.

\subsection{Partiality and the additive affine completion}

What is the appropriate notion of \emph{partiality} for a given rig category? If coproducts and a terminal object are available the \emph{lift monad} $(-)+1$ can answer this question, but not every rig category has these. Instead, we can think of a partial map $A \to B$ as a map $A \to B \oplus E$, \textit{i.e.}, extend the output state space with an extra part $E$ to receive all the inputs we wish to be undefined. Any map $f \colon A \to B$ can be lifted
to a total map $\rho_\oplus^{-1} \circ f \colon A \to B \oplus O$ using the inverse right unitor, and partial maps $f \colon A \to B
\oplus E$ and $g \colon B \to C \oplus E'$ can be composed by composing their defined parts, i.e., $\alpha_\oplus \circ g \otimes \id_E \circ f \colon A \to C
\oplus (E' \oplus E)$. This is an information-preserving variant of Kleisli-composition for the lift monad.

This describes the \emph{additive affine completion} of a rig category, barring one detail: given that $E$ describes where the map is undefined, it shouldn't actually matter how we represent this particular part. For example, given some partial map $f \colon A \to B \oplus E$ and some manipulation $m \colon E \to E'$ of the undefined part, $f$ and $(\id \oplus m) \circ f$ morally describe the same partial map. Therefore, given morphisms $f \colon A \to B \oplus E $ and $f' \colon A \to B \oplus E'$, we write $f \le_L f'$ if and only if there exists some \emph{mediator} $m \colon E \to E'$ such that
\[\begin{tikzcd}[ampersand replacement=\&]
	\& A \\
	{B \oplus E} \&\& {B \oplus E'}
	\arrow["f"', from=1-2, to=2-1]
	\arrow["{f'}", from=1-2, to=2-3]
	\arrow["{\id \oplus m}"', from=2-1, to=2-3]
\end{tikzcd}\]
commutes. This straightforwardly gives a preorder, though not (necessarily) an equivalence relation since mediators need not be invertible. However, since we would like momentarily to treat it as an equivalence, we consider instead its equivalence closure $\sim_L$, \textit{i.e.},, the least equivalence relation containing $\le_L$.
\begin{definition}\label{def-affine-comp}
  Given a rig category $\cat{C}$, its \emph{additive affine completion} 
  $\La(\cat{C})$ is the category whose
  \begin{itemize}
    \item objects are those of $\cat{C}$,
    \item morphisms $[f,E] \colon A \to B$ are pairs of an object $E$ and an
    equivalence class of morphisms $f \colon A \to B \oplus E$ of $\cat{C}$
    under $\sim_L$,
    \item identities $A \to A$ are $[\rho_\oplus^{-1}, O]$ (with $\rho_\oplus
    \colon A \oplus O \to A$ the right unitor), and
    \item composition of $[f,E] \colon A \to B$ and $[g,E'] \colon B \to C$ is 
    $[\alpha_\oplus \circ g \otimes \id_E \circ f, E' \oplus E]$.
  \end{itemize}
\end{definition}
There is a dual to this construction, the \emph{additive coaffine completion} $\Ra(\cat{C})$, defined as $\La(\cat{C}^\opp)^\opp$. Explicitly, morphisms $A \to B$ in $\Ra(\cat{C})$ are equivalence classes of morphisms $A \oplus E \to B$, and so hide part of their \emph{source} space rather than their \emph{target} space. We summarise some features of these categories.

\begin{proposition}
  When $\cat{C}$ is a rig category, so are $\La(\cat{C})$ and $\Ra(\cat{C})$.
\end{proposition}
\begin{proof}
  That $\Ra(\cat{C})$ is a rig category was shown in \cite[Lemma
  12]{heunenkaarsgaard:qie}; that $\La(\cat{C})$ is also a rig category follows
  by $\La(\cat{C}) \cong \Ra(\cat{C}^\opp)^\opp$ and the fact that $\cat{C}$ is
  a rig category iff $\cat{C}^\opp$ is.
\end{proof}
\begin{proposition}
  The additive unit $O$ is terminal in $\La(\cat{C})$ and initial in 
  $\Ra(\cat{C})$.
\end{proposition}
\begin{proof}
  The inverse left unitor $\lambda_\oplus^{-1} \colon A \to O \oplus A$ of $\cat{C}$ represents a morphism $A \to O$ in $\La(\cat{C})$; this satisfies the universal property of the terminal object by definition of $\sim_L$. Dually, $O$ is initial in $\Ra(\cat{C})$.
\end{proof}

We call a rig category \emph{additively coaffine} when the additive unit is initial, and \emph{additively affine} when it is terminal.

\begin{proposition}
  There are strict rig functors $\mathcal{D} \colon \cat{C} \to \La(\cat{C})$ 
  and $\mathcal{E} \colon \cat{C} \to \Ra(\cat{C})$.
\end{proposition}
\begin{proof}
  Define $\mathcal{D}(A) = A$ on objects, and $\mathcal{D}(f) =
  [\rho_\oplus^{-1} \circ f, O]$ on morphisms, and $\mathcal{E}$ dually.
  Straightforward calculations show that this defines strict rig functors.
\end{proof}

As the name suggests, these are, indeed, completions on rig categories.
\begin{proposition}\label{prop-univ}
  $\La(\cat{C})$ is the additive affine completion of $\cat{C}$ in the 
  following sense: given any additively affine rig category $\cat{D}$ and 
  strong rig functor $F \colon \cat{C} \to \cat{D}$, there is a unique strong rig 
  functor $\hat{F} \colon \La(\cat{C}) \to \cat{D}$ making the diagram below commute.
  \[\begin{tikzcd}[ampersand replacement=\&]
  	{\cat{C}} \& {\La(\cat{C})} \\
  	\& {\cat{D}}
  	\arrow["{\mathcal{D}}", from=1-1, to=1-2]
  	\arrow["{\hat{F}}", from=1-2, to=2-2]
  	\arrow["F"', from=1-1, to=2-2]
  \end{tikzcd}\]
\end{proposition}
\begin{proof}
  By \cite[Theorem 19]{heunenkaarsgaard:qie} and duality.
\end{proof}

Dualising this result exhibits $\Ra(\cat{C})$ as the additive coaffine completion of $\cat{C}$.

\section{Dagger partiality and the additive biaffine completion}
\label{sec:biaffine_completion}

As the (co)affine completion explicitly involves hiding a part of the source or target space of a morphism, we cannot expect to lift either construction to a completion of \emph{dagger} rig categories. A great example of this fact is demonstrated by considering $\Ra(\Unitary)$. Write $\cat{Isometry}$ for the category of finite-dimensional Hilbert spaces and morphisms satisfying $f^\dag \circ f = \mathrm{id}$, and $\cat{coIsometry}$ for its dual.

\begin{proposition}\label{prop-r-unitary}
  There are rig equivalences $\Ra(\Unitary) \simeq \Isometry$ and $\La(\Unitary) \simeq \coIsometry$.
\end{proposition}
\begin{proof}
  That $\Ra(\Unitary) \simeq \Isometry$ was shown in \cite{huotstaton:universal}. The other statement follows from duality: $\La(\Unitary) 
  \simeq \Ra(\Unitary^\opp)^\opp \simeq \Ra(\Unitary)^\opp \simeq \Isometry^\opp 
  \simeq \coIsometry$.
\end{proof}
However, though $\Unitary$ is a dagger rig category, $\Isometry$ and $\coIsometry$ are mere rig categories. 
Intuitively, this must be the case because dagger categories are self-dual (\textit{i.e.}, satisfy $\cat{C} \cong \cat{C}^\opp$) so limits and colimits coincide, but the affine completion only adds (certain) limits without the corresponding colimits. However, this also suggests that if we seek a notion of partiality that respects daggers, we would need the additive unit to be both initial and terminal, \textit{i.e.}, a \emph{zero object}. Fortunately, we can ensure this by applying \(\La\) after \(\Ra\) or vice versa.

\begin{proposition}
  The additive unit $O$ is a zero object in both $\La(\Ra(\cat{C}))$ and 
  $\Ra(\La(\cat{C}))$.
\end{proposition}
\begin{proof}
  By \cite[Lemma 11]{heunenkaarsgaard:qie}, $\Ra(-)$ preserves terminal
  objects, and by duality, $\La(-)$ preserves initial objects. Thus $O$ is both 
  initial and terminal (\textit{i.e.}, a zero object) in both $\La(\Ra(\cat{C}))$ and 
  $\Ra(\La(\cat{C}))$.
\end{proof}

The situation is interesting: neither $\La(-)$ nor $\Ra(-)$ on their own preserve dagger rig categories, but as we will see, their combination does. Hence it is advantageous to consider them together for dagger rig categories, which also leads to a slightly simpler presentation.
In a rig category, define $\medeq{LR}$ as the least equivalence relation containing the three relations $\medeq{\id}$, $\medeq{L}$, and $\medeq{R}$
defined as follows, for all $f \colon A \oplus H \to B \oplus G$:
\begin{itemize}
  \item $f \medeq{L} (\id \oplus m) \circ f$ for all $m \colon G \to G'$;
  \item $f \medeq{R} f \circ (\id \oplus n)$ for all $n \colon H' \to H$;
  \item $f \medeq{\id} \alpha_\oplus \circ (f \oplus \id_X) \circ 
  \alpha_\oplus^{-1}$ for all identities $\id_X$.
\end{itemize}

\begin{definition}
  Given a rig category $\cat{C}$, its \emph{biaffine completion} 
  $\LRa(\cat{C})$ is the category whose
  \begin{itemize}
    \item objects are those of $\cat{C}$,
    \item morphisms $[H,f,G] \colon A \to B$ are triples consisting of two objects 
    $H$ and $G$ and a morphism $A \oplus H \to B \oplus G$ of $\cat{C}$ 
    quotiented by $\medeq{LR}$,
    \item identities $A \to A$ are $[O,\id_{A \oplus O},O]$, and
    \item the composition of $f \colon A \oplus H \to B \oplus G$ and 
    $g \colon B \oplus H' \to C \oplus G'$ is given by
  \end{itemize}
  $$A \oplus (H \oplus H') \tot{\alpha_\oplus^{-1}} 
  (A \oplus H) \oplus H' \tot{f \oplus \id}
  (B \oplus G) \oplus H' \tot{\cong}
  (B \oplus H') \oplus G \tot{g \oplus \id}
  (C \oplus G') \oplus G \tot{\alpha_\oplus}
  C \oplus (G' \oplus G).$$
\end{definition}

We state some properties of the $\LRa$-construction, the proofs of which can be found in the appendix.

\begin{proposition}\label{prop-lra-zero}
  The additive unit $O$ is a zero object in $\LRa(\cat{C})$.
\end{proposition}

\begin{proposition}\label{prop-lra-dagger}
  When $\cat{C}$ is a dagger rig category, so is $\LRa(\cat{C})$.
\end{proposition}

As before, there is a functor $\mathcal{F} \colon \cat{C} \to \LRa(\cat{C})$ given on objects by $\mathcal{F}(A) = A$ and on morphisms by $F(f) = [0, f \oplus \id_0, 0]$, making $\LRa(\cat{C})$ a completion in the formal sense. We say that a rig category is \emph{additively biaffine} if the unit of the sum is a zero object.
 
\begin{theorem}\label{thm-add-biaff-compl}
  $\LRa(\cat{C})$ is the additive biaffine completion of $\cat{C}$ in the 
  following sense: given any additively biaffine rig category $\cat{D}$ and 
  strong rig functor $F \colon \cat{C} \to \cat{D}$, there is a unique strong rig 
  functor $\hat{F} \colon \LRa(\cat{C}) \to \cat{D}$ making the diagram below 
  commute.
  \[\begin{tikzcd}[ampersand replacement=\&]
  	{\cat{C}} \& {\LRa(\cat{C})} \\
  	\& {\cat{D}}
  	\arrow["{\mathcal{F}}", from=1-1, to=1-2]
  	\arrow["{\hat{F}}", from=1-2, to=2-2]
  	\arrow["F"', from=1-1, to=2-2]
  \end{tikzcd}\]
\end{theorem}
\begin{proof}
  Assuming such a functor exists, we first prove its uniqueness. Imposing \(F = \widehat{F} \circ \cF\) implies that \(\widehat{F}(A) = F(A)\) on objects and, on morphisms in the image of \(\cF\), we have that \(\widehat{F}(\cF(f)) = F(f)\). It is easy to check (for instance, diagrammatically) that any \([H,f,G] \colon A \to B\) in \(\LRa[\C]\) decomposes as follows:
  \begin{equation*}
    [H,f,G] = A \tot{\rho^{-1}} A \oplus O \tot{\id \oplus !} A \oplus H \tot{\cF(f)} B \oplus G \tot{\id \oplus !} B \oplus O \tot{\rho} B
  \end{equation*}
  where the objects and morphisms shown are in \(\LRa[\C]\), where \(O\) is a zero object, and \(!\) refer to the corresponding unique morphisms.
  The image under \(\widehat{F}\) of each of these morphisms is uniquely determined because \(\widehat{F}\) is assumed to be monoidal, the fact that \(\D\) has \(0 \cong F(O)\) as its zero object, and the equality \(\widehat{F}(\cF(f)) = F(f)\) discussed above.
  Thus it only remains to prove that \(\widehat{F}\) as defined above is indeed a strong rig functor.
  By definition, \(\widehat{F}(\id) = F(\id)\), and functoriality is easy to check using naturality of \(\rho\) and the fact that \(0\) is a zero object. The fact that \(\widehat{F}\) is a strong rig functor follows directly from the same property for \(F\), since \(\cF\) is a strict rig functor, \(\widehat{F}(O) = F(O) \cong 0\), and \(\widehat{F}(A \oplus B) = F(A \oplus B) \cong F(A) \oplus F(B) = \widehat{F}(A) \oplus \widehat{F}(B)\).
\end{proof}
Using this characterisation, we can show commutativity of the $\La$ and $\Ra$ constructions, the proof of which is found in the appendix.
\begin{proposition}\label{prop-LR-eq}
  When $\C$ is a rig category, \(\La(\Ra(\C)) \cong \LRa(\C) \cong \Ra(\La(\C))\).
\end{proposition}

\subsection{Unitaries with dagger partiality are contractions}

A linear map $f \colon A \to B$ between normed spaces $A$ and $B$ is \emph{(weakly) contractive} iff $\norm{f(x)} \le \norm{x}$ for all $x \in A$. So contractive maps include all isometries, coisometries, and unitaries.
We now show that applying the $\LRa{}$-construction to the category of finite-dimensional Hilbert spaces and unitaries is equivalent to the category $\Contraction$ of finite-dimensional Hilbert spaces and contractions. Another way of saying this is that extending the unitaries with dagger partiality gives precisely the contractions.

\begin{theorem}\label{thm-lra-unitary-contraction}
  $\LRa(\Unitary) \cong \Contraction$.
\end{theorem}
\begin{proof}
  Proposition~\ref{prop-LR-eq} gives $\LRa(\Unitary) \cong
  \La(\Ra(\Unitary))$, and Proposition~\ref{prop-r-unitary} gives
  $\Ra(\Unitary) \cong \Isometry$, so it suffices to show $\La(\Isometry) \cong \Contraction$.
  The strategy is to define a functor \(F \colon \Contraction \to \La(\Isometry)\) and prove it is full, faithful and essentially surjective.
  On objects, \(F(A) = A\) so \(F\) is essentially surjective.
  Let \(T \colon A \to B\) be a contraction and let \(f \colon A \to B \oplus G\) be an isometry such that \(T = \pi_B \circ f\): for example, by \cite{halmos:normal,robinson:julia}, an isometric dilation $A \to B \oplus A$ can be constructed as
  $$
  f = \begin{pmatrix}
    T \\ (1 - T^\dagger T)^\frac{1}{2}
  \end{pmatrix}.
  $$
  Let \(F(T) = [f,G]\).
  We need to check this mapping is well-defined, \textit{i.e.}, if \(f' \colon A \to B \oplus G'\) also satisfies \(T = \pi_B \circ f'\), verify that \([f,G] = [f',G']\).
  To do this, first consider the Hilbert space \(\im(f) \subseteq B \oplus G\) and similarly \(\im(f')\) and define a function \(g \colon \im(f') \to \im(f)\) by \(f'(a) \mapsto f(a)\).
  This function is well-defined because $f$ and $f'$ are isometries and hence injective.
  It is easy to see that \(g\) is linear; also notice that
  $$
    \braket{g(f'a)}{g(f'a)} = \braket{fa}{fa} = \braket{a}{a} = \braket{f'a}{f'a}
  $$
  because \(f\) and \(f'\) are isometries.
  This applies to all vectors in \(\im(f')\), so \(g\) is an isometry; in fact, because \(\dim(\im(f)) = \dim(\im(f'))\) is finite, \(g\) has an inverse that is also an isometry, so \(g \colon \im(f') \to \im(f)\) is unitary.
  Whenever \(f(a) \in B\) it is necessary that \(f(a) = f'(a)\) for them to be dilations of \(T\), which means that \(g\) acts as the identity on \(\im(T)\). Since \(\im(f) = (\im(f) \cap B) \oplus (\im(f) \cap G)\) and similarly for \(\im(f')\), while \(\im(f) \cap B = \im(T) = \im(f') \cap B\), we can decompose \(g\) into a block matrix of the form \(g \colon \im(T) \oplus (\im(f') \cap G') \to \im(T) \oplus (\im(f) \cap G)\). Given that the component \(\im(T) \to \im(T)\) has been established to be the identity, and \(g\) is unitary, it follows that \(g\) must be of the form
  $$
     g = \begin{pmatrix}
       \id & 0 \\
       0 & h
     \end{pmatrix}
  $$
  for some unitary \(h \colon \im(f') \cap G' \to \im(f) \cap G\). Next, we lift \(h\) to a map \(k \colon G' \to G\). Assume without loss of generality that \(\dim(G') \leq \dim(G)\) and pick an isometry \(k \colon G' \to G\) satisfying \((\id_B \oplus k) \circ f' = f\).
  Such a function exists: let \(k(x) = h(x)\) whenever \(x \in \im(f') \cap G'\) and, for each element \(x \in G'\) not in \(\im(f')\), choose an element \(y \in G\) not in \(\im(f)\) and let \(k(x) = y\).
  Since \(\dim(G') \leq \dim(G)\) by assumption, the latter choices can be made so that \(k\) is an isometry.
  Then, it is immediate that \(f \medeq{L} f'\) since \(k \colon G' \to G\) is an isometry acting as their mediator, and we conclude that \([f,G] = [f',G']\).
  Moreover, if \(F(T) = [f,G]\) and \(F(S) = [g,G']\) then it's easy to check that \((g \oplus \id_G)\circ  f\) is an isometric dilation of \(S\circ T\), so \(F\) is indeed a functor.
  For any two contractions \(T\) and \(S\), if \(F(T) = F(S)\) then there is an isometric dilation \(f\) such that \(T = \pi \circ f = S\) and, hence, \(F\) is faithful.
  Finally, \(\pi_G \circ f\) is a contraction for all \([f,G]\), so \(F\) is full.
\end{proof}

\section{Hiding in a partial setting}
\label{sec:generalized_pablo_pushout}

The previous section considered (co)affine completions with respect to the direct sum. But rig categories have another monoidal structure, the tensor product. Just as we can form the \emph{additive} affine completion $\La(\cat{C})$ of a rig category $\C$, we can also form the \emph{multiplicative} one $\Lm(\cat{C})$: objects are those of $\cat{C}$, morphisms $A \to B$ are equivalence classes of morphisms $A \to B \otimes G$ in $\cat{C}$, with identities and composition just as in Definition~\ref{def-affine-comp}. Where the $\La$-construction captures \emph{partiality} by allowing part of the \emph{state space} to be hidden, the $\Lm$-construction models \emph{discarding} by allowing any \emph{state} to be hidden (partly or fully).

This construction was shown to capture Stinespring dilation in \cite{huotstaton:universal}, where the authors argued that $\Lm(\Isometry)$ is monoidally equivalent to the category $\FHilbCPTP$ of finite-dimensional Hilbert spaces and completely positive trace-preserving (CPTP) maps (or \emph{quantum channels}). A drawback is that this does not generally preserve the direct sum, so $\Lm(\cat{C})$ is generally only monoidal (under tensor product) when $\cat{C}$ is a rig category (though remnants of the direct sum persist, see \cite[Section~4.4]{heunenkaarsgaard:qie}).

Unfortunately, we cannot hope to use the $\Lm$-construction to construct the category $\FHilbCPTN$ of finite-dimensional Hilbert spaces and completely positive trace-nonincreasing (CPTN) maps (or \emph{partial quantum channels}) for a simple reason: the unit $I$ of the tensor product is terminal in $\Lm(\cat{C})$, but unlike $\FHilbCPTP$, it is not terminal in $\FHilbCPTN$. The unique map $A \to I$ in $\FHilbCPTP$ is given by the trace of a density matrix on $A$, and it is unique because all density matrices have unit trace. On the other hand, the \emph{subnormalised} density matrices found in $\FHilbCPTN$ take their traces in the unit interval, so even though there is only one \emph{trace-preserving} map $A \to I$, there are as many trace-nonincreasing ones as there are real numbers in $[0,1]$. Another way to say this is that the $\Lm$-construction fails to respect \emph{partiality} in $\FHilbCPTN$ (as also discussed in \cite{heunenkaarsgaard:bennett}).

This section generalises the result from \cite{huotstaton:universal} to the partial case. To do this, we show a variant of Stinespring dilation for completely positive trace-nonincreasing maps, describe a construction $\Lt(\cat{C})$ that extends $\cat{C}$ with unique \emph{total} deletion maps, and relate the two by showing that $\Lt(\Contraction) \simeq \FHilbCPTN$. We relate this to the total case by showing that $\Lt(\cat{C})$ has a universal property as a certain pushout in the category of (locally small) monoidal categories.

\subsection{Stinespring dilation for partial quantum channels}

We begin with a small lemma, which turns out to be incredibly useful when working with contractions.

\begin{proposition}\label{prop-contraction}
  $T$ is contractive if and only if $T^\dagger T \le 1$.
\end{proposition}
\begin{proof}
  By definition, $T^\dagger T \le 1$ iff 
  $1 - T^\dagger T$ is positive semidefinite, which 
  in turn is the case iff
  $\bra{\phi}(1 - T^\dagger T)\ket{\phi} \ge 0$  for all $\ket\phi$.
  But since
  $$\bra{\phi}(1 - T^\dagger T)\ket{\phi} =
  \bra{\phi}\ket{\phi} - \bra{\phi}T^\dagger T\ket{\phi} = \norm{\ket\phi}^2 -
  \norm{T \ket{\phi}}^2\text,$$
  $1 - T^\dagger T$ is positive semidefinite iff
  $\norm{\ket\phi}^2 - \norm{T \ket{\phi}}^2 \ge 0$ for all $\ket\phi$, \textit{i.e.}, when $T$ is contractive.
\end{proof}

The previous proposition links handily to the following theorem about the Kraus representation of completely positive trace-nonincreasing maps.

\begin{proposition}[\cite{cappellinisommerszyczkowski:subcptn}]
  Any CPTN map $\Phi$ admits a representation
  $\Phi(\rho) = \sum_{i=1}^k M_i \rho M_i^\dagger$ with $\sum_i
  M_i^\dagger M_i \le 1$.
\end{proposition}

As in the trace-preserving case, we can construct a Stinespring dilation from the Kraus representation.

\begin{proposition} \label{prop-stinespring}
  Every CPTN map $\Phi$ admits a Stinespring dilation 
  $\Phi(\rho) = \tr_E(T \rho T^\dagger)$ for some contraction $T$ and Hilbert
  space $E$.
\end{proposition}
\begin{proof}
  Let $\B(A) \tot{\Phi} \B(B)$ be a partial quantum channel with Kraus
  representation $\Phi(\rho) = \sum_{i=1}^k M_i \rho M_i^\dagger$. Define $E =
  \mathbb{C}^k$, let $A \otimes \mathbb{C} \tot{U_R} A$ denote the right
  unitor, and let $T = \left(\sum_{i=1}^k M_i \otimes
  \ket{i}\right) U_R^\dagger$. Now for any $\rho$:
  \begin{align*}
    \tr_E(T \rho T^\dagger) & = \tr_E\left(\left(\sum_{i=1}^k M_i \otimes
    \ket{i}\right) U_R^\dagger \, \rho \, U_R \left(\sum_{j=1}^k M_j^\dagger 
    \otimes
    \bra{j}\right)\right) 
    = \tr_E\left(\sum_{i=1}^k \sum_{j=1}^k M_i \rho M_j^\dagger \otimes 
    \dyad{i}{j}\right)  \\
    & = \sum_{i=1}^k \sum_{j=1}^k \tr(\dyad{i}{j}) M_i \rho M_j^\dagger
    = \sum_{i=1}^k M_i \rho M_i^\dagger = \Phi(\rho) \enspace.
  \end{align*}
  It remains to show that $T$ is contractive:
  \begin{align*}
    T^\dagger T & = U_R \left(\sum_{j=1}^k M_j^\dagger \otimes \bra{j}\right)
    \left(\sum_{i=1}^k M_i \otimes \ket{i}\right) U_R^\dagger
    = U_R \left(\sum_{j=1}^k \sum_{i=1}^k M_j^\dagger M_i \otimes
    \bra{j}\ket{i} \right) U_R^\dagger \\
    & = U_R \left(\sum_{i=1}^k M_i^\dagger M_i \otimes \bra{i}\ket{i} \right)
    U_R^\dagger
    = \sum_{i=1}^k M_i^\dagger M_i \le 1
  \end{align*}
  which finishes the proof.
\end{proof}

This representation is essentially unique, \textit{i.e.}, unique up to an isometry applied to the ancilla.

\begin{lemma}[{\cite[Theorem 8.2 ]{nielsen_chuang_2010}}]
  Let $\Phi$ be a CP map with two Kraus representations
  $\Phi(\rho) = \sum_{i=1}^k E_i \rho E_i^\dagger$ and $\Phi(\rho) =
  \sum_{j=1}^{k'} F_j \rho F_j^\dagger$. Assume $k = k'$; if $k < k'$, add some $E_i = 0$ for all $k < i \leq k'$. There is a unitary $k$-by-$k$ matrix $U =
  (u_{ij})$ such that $E_i = \sum_{j=1}^k u_{ij} F_j$ for
  all $i$.
\end{lemma}

\begin{proposition} \label{prop-stinespring-unique}
  Let $\Phi$ be a CPTN map, and let $T_E \colon A \to B \otimes G$, $T_F \colon
  A \to B \otimes G'$ be two contractions such that $\Phi(\rho) = \tr_G(T_E
  \rho T_E^\dagger) = \tr_{G'}(T_F \rho T_F^\dagger)$. If $\dim(G) \leq
  \dim(G')$, then there is an isometry $W \colon G \to G'$ such that $T_F =
  (1 \otimes W) T_E$.
\end{proposition}
\begin{proof}
  Write $U_R$ for the right unitor.
  Fix an orthonormal basis for $G$ and define $E_i = U_R (1 \otimes \bra{i}) T_E$; do the same for $F_j = U_R (1 \otimes \bra{j}) T_F$.
  Notice that $T_E = (\sum_i E_i \otimes \ket{i}) U_R^\dagger$ thanks to $\sum_i \dyad{i}{i} = 1_G$ for an orthonormal basis.
  Take $V$ to be an injection of $G$ into a Hilbert space with the same dimension as $G'$. Let $U = (u_{ij})$ be the unitary from the lemma above. Then:
  $$
    T_E = (\sum_i E_i \otimes \ket{i})U_R^\dagger = (\sum_i (\sum_j u_{ij} F_j) \otimes \ket{i})U_R^\dagger 
     = (\sum_i (\sum_j u_{ij} F_j \otimes \ket{i}))U_R^\dagger = (\sum_j F_j \otimes (\sum_i u_{ij} \ket{i}))U_R^\dagger.
  $$
  Because $U$ is unitary, $\ket{\hat{j}} = \sum_i u_{ij} \ket{i}$ form an orthonormal basis for $G'$.
  Finally, we obtain the isometry $W \colon G \to G'$ by composing $W = R U V$ where $R$ is the unitary mapping $\ket{\hat{j}} \mapsto \ket{j}$.
\end{proof}

Moreover, as in the trace-preserving case, contractions give rise to CPTN maps through conjugation.

\begin{proposition}\label{prop-conj-cptn}
  Every contraction $T$ gives rise to a CPTN map $\Phi(\rho) = T \rho 
  T^\dagger$.
\end{proposition}
\begin{proof}
Conjugation by any linear map is completely positive. That $\Phi$ is trace-nonincreasing follows by $\tr(T \rho T^\dagger) =
  \tr(\rho T^\dagger T) \le \tr(\rho \, 1) = \tr(\rho)$.
\end{proof}

Taken together, these results show that the situation between $\Contraction$ and $\FHilbCPTN$ mirrors that between $\Isometry$ and $\FHilbCPTP$: there is a (strict monoidal) functor $\Contraction \to \FHilbCPTN$ that sends a contraction $T$ to conjugation by $T$, and every CPTN map can be expressed this way (in an essentially unique way) if we allow ourselves an ancilla system.

\subsection{Total eclipse of the state}

We now formulate the new $\Lt$-construction, adding a notion of hiding that cooperates with the \emph{total} maps in a monoidal dagger category. Consider a monoidal dagger category $\cat{C}$. We can think of its dagger monomorphisms (\textit{i.e.}, morphisms $f$ satisfying $f^\dagger \circ f = \id$) as its \emph{total maps}. These form a subcategory of $\cat{C}$, which we denote $\DagMon(\cat{C})$. The significance of the total maps is that, instead of being able to mediate with \emph{arbitrary} maps from $\cat{C}$ (as in the $\Lm$-construction), we are only permitted to mediate with \emph{total} maps (\textit{i.e.}, morphisms in $\DagMon(\cat{C})$) in the $\Lt$-construction. Explicitly, for $f \colon A \to B \otimes G$ and $f' \colon A \to B \otimes G'$, write $f \le_{\Lt} f'$ iff there exists a dagger monomorphism $m \colon G \to G'$ making
\[\begin{tikzcd}[ampersand replacement=\&]
	\& A \\
	{B \otimes G} \&\& {B \otimes G'}
	\arrow["{\id \otimes m}"', from=2-1, to=2-3]
	\arrow["f"', from=1-2, to=2-1]
	\arrow["{f'}", from=1-2, to=2-3]
\end{tikzcd}\]
commute. Let $\sim_{\Lt}$ denote the equivalence closure of $\le_{\Lt}$.
\begin{definition}
  The \emph{partial multiplicative affine completion} $\Lt(\cat{C})$ of a 
  symmetric monoidal dagger category $\cat{C}$ is the category whose
  \begin{itemize}
    \item objects are those of $\cat{C}$,
    \item morphisms $[f,G] \colon A \to B$ are pairs of an object $G$ and an
    equivalence class of morphisms $f \colon A \to B \otimes G$ of $\cat{C}$
    under $\sim_{\Lt}$,
    \item identities are $[\rho^{-1}_\otimes, I]$, and
    \item composition of $[f,G] \colon A \to B$ and $[g,G'] \colon B \to C$ is
    $[\alpha_\otimes \circ g \otimes \id_G \circ f, G' \otimes G]$.
  \end{itemize}
\end{definition}

Before proceeding with the universal property of this construction, we first establish that it succeeds in constructing $\FHilbCPTN$.

\begin{theorem}\label{thm-lt-contr}
  There is an equivalence $\Lt(\Contraction) \simeq \FHilbCPTN$ of monoidal categories.
\end{theorem}
\begin{proof}
  Construct a functor $\Lt(\Contraction) \to \FHilbCPTN$ acting as the
  identity on objects, by sending $[f,G] \colon A \to B$ to the map $\rho \mapsto
  \tr_G(f^\dagger \rho f)$, which is CPTN by Proposition~\ref{prop-conj-cptn}
  and since the partial trace is trace-preserving. This is well-defined since 
  dagger monomorphisms in $\Contraction$ are precisely the isometries, and since
  Stinespring dilations are invariant under isometric manipulation of the 
  ancilla system $G$. This functor is essentially surjective since it is
  identity on objects and $\Lt(\Contraction)$ and $\FHilbCPTN$ have the same
  objects; it is full since every CPTN map admits a Stinespring dilation (by 
  Proposition~\ref{prop-stinespring}); it is faithful since different
  Stinespring dilations of the same CPTN map are always connected by an 
  isometry on the ancilla by Proposition~\ref{prop-stinespring-unique}; and it 
  is (strict) monoidal since it preserves coherence isomorphisms.
\end{proof}

An immediate consequence of the definition of $\Lt(\cat{C})$ is that each object comes equipped with a discarding map and chosen projections.

\begin{proposition}
  Every object $A$ of $\Lt(\cat{C})$ has a discarding map $\discard : A \to I$, giving canonical projections $\pi_1 \colon A \otimes B \to A$ and $\pi_2 \colon A \otimes B \to B$.
\end{proposition}
\begin{proof}
  As in $\La$, construct $\discard : A \to I$ as the equivalence class
  of the pair $[\lambda_\otimes^{-1},A]$ where $\lambda_\otimes^{-1} \colon A \to I \otimes A$ is the inverse left unitor of $\cat{C}$. Projections are given by the equivalence class of $[\id_{A \otimes B},B] \colon A \otimes B \to A$ and that of $[\sigma_\otimes,A] \colon A \otimes B \to B$.
\end{proof}

We state some basic properties of this construction, shown analogously to those for the $\La$-construction.

\begin{proposition}
  When $\cat{C}$ is a symmetric monoidal dagger category, $\Lt(\cat{C})$ is 
  symmetric monoidal.
\end{proposition}

\begin{proposition}
  There is a strict monoidal functor $\mathcal{E}_t \colon \cat{C} \to \Lt(\cat{C})$ 
  given by $\mathcal{E}_t(A) = A$ on objects and $\mathcal{E}_t(f) = 
  [\rho^{-1}_\otimes \circ f, I]$ on morphisms.
\end{proposition}

The connection between the $\Lm$ and the $\Lt$ construction is made clear by the following inclusion.

\begin{proposition}
  There is a strict monoidal inclusion functor $I_t \colon \Lm(\DagMon(\cat{C})) \to 
  \Lt(\cat{C})$.
\end{proposition}

Indeed, as we will see momentarily, the $\Lt$-construction is characterised by this inclusion.
Interestingly, this construction can be seen as a particular instance of the \emph{monoidal indeterminates}-construction~\cite{hermidatennent:indeterminates}. This gives a very useful factorisation lemma, which may be regarded as an instance of \emph{purification}.

\begin{lemma}
  Let $\cat{C}$ be a symmetric monoidal dagger category. Then
  \begin{enumerate}[(i)]
    \item morphisms $[f,E]$ of $\Lm(\DagMon(\cat{C}))$ factor uniquely as $\pi_1
    \circ \mathcal{E}(f')$ for some $f'$ in $\DagMon(\cat{C})$,
    \item morphisms $[f,E]$ of $\Lt(\cat{C})$ factor uniquely as $\pi_1
    \circ \mathcal{E}_t(f')$ for some $f'$ in $\cat{C}$, and
    \item morphisms $[f,E]$ of $\Lt(\cat{C})$ factor uniquely as
    $I_t(\pi_1) \circ \mathcal{E}(f')$ for some $f'$ in $\cat{C}$.
  \end{enumerate}
\end{lemma}
\begin{proof}
  (i,ii) are the expansion-raw factorisation of~\cite{hermidatennent:indeterminates}, and (iii) holds as $\pi_1 = I_t(\pi_1)$ in $\Lt(\cat{C})$.
\end{proof}

We are now ready to establish the universal property of $\Lt(\cat{C})$. 

\begin{theorem}\label{thm-pushout}
  $\Lt(\cat{C})$ is a pushout of $\mathcal{E} \colon \DagMon(\cat{C}) \to
  \Lm(\DagMon(\cat{C}))$ along the inclusion functor $\DagMon(\cat{C})
  \rightarrowtail \cat{C}$ in the category of locally small symmetric monoidal 
  categories and strong monoidal functors.
  
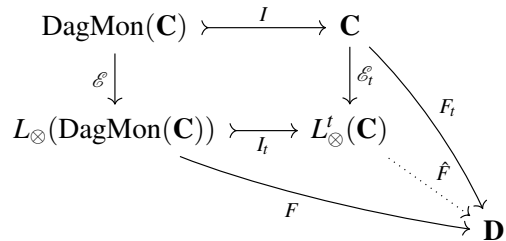
\[\begin{tikzcd}[ampersand replacement=\&]
  	{\DagMon(\mathbf{C})} \& {\mathbf{C}} \\
  	{\Lm(\DagMon(\mathbf{C}))} \& {\Lt(\mathbf{C})} \\
  	\&\& {\mathbf{D}}
  	\arrow["I", tail, from=1-1, to=1-2]
  	\arrow["{\mathcal{E}}"', from=1-1, to=2-1]
  	\arrow["{I_t}"', tail, from=2-1, to=2-2]
  	\arrow["{\mathcal{E}_t}", from=1-2, to=2-2]
  	\arrow["{F}"', curve={height=12pt}, from=2-1, to=3-3]
  	\arrow["F_t", curve={height=-12pt}, from=1-2, to=3-3]
  	\arrow["{\hat{F}}", dotted, from=2-2, to=3-3]
  \end{tikzcd}\]
\end{theorem}
\begin{proof}[Proof sketch]
  Notice first that the upper square commutes since all functors involved
  are strict monoidal, and $I$ and $I_t$ are merely inclusions behaving as the
  identity on objects and morphisms, while $\mathcal{E}$ is precisely
  $\mathcal{E}_t$ restricted to dagger monomorphisms of $\cat{C}$.
  
  Next, since objects on $\cat{C}$, $\DagMon(\cat{C})$, $\Lt(\cat{C})$, and
  $\Lm(\DagMon(\cat{C}))$ all coincide, $F$ and $F_t$ must agree on objects, so 
  we may
  define $\hat{F}(X) = F(X) = F_t(X)$ on objects, and $\hat{F} \circ I_t = F$
  and $\hat{F} \circ \mathcal{E}_t = F_t$ on objects follows immediately. On
  morphisms we define
  $\hat{F}([f,E]) = F(\pi_1) \circ F_t(f)$. That this definition satisfies 
  $\hat{F} \circ I_t = F$ and $\hat{F} \circ \mathcal{E}_t = F_t$ on morphisms 
  as well, and that $\hat{F}$ is unique, follows using the factorisation lemma 
  above (full proof in the appendix).
\end{proof}

Instantiating this property to the case of $\Contraction$, where the functors $\Contraction \to \Isometry$ and $\FHilbCPTP \to \FHilbCPTN$ are inclusions, and the functors $\Isometry \to \FHilbCPTP$ and $\Contraction \to \FHilbCPTN$ conjugating by an isometry or contraction, we obtain the following characterisation.

\begin{corollary}
  The following square is a pushout of symmetric monoidal categories:
  \[\begin{tikzcd}[ampersand replacement=\&]
  	\Isometry \& \Contraction \\
  	\FHilbCPTP \& \FHilbCPTN
  	\arrow[from=1-1, to=2-1]
  	\arrow[from=1-1, to=1-2]
  	\arrow[from=1-2, to=2-2]
  	\arrow[from=2-1, to=2-2]
  \end{tikzcd}\]
\end{corollary}
\begin{proof}
  That $\DagMon(\Contraction) \simeq \Isometry$ follows by the fact that every 
  isometry $f$ is contractive (it satisfies $\norm{f(x)} = \norm{x}$, so 
  specifically $\norm{f(x)} \le \norm{x}$) and satisfies $f^\dagger \circ f = 
  \id$. Finally, $\Lm(\Isometry) \simeq \FHilbCPTP$ is known~\cite{huotstaton:universal}, and $\Lt(\Contraction) \simeq \FHilbCPTN$
  by Theorem~\ref{thm-lt-contr}.
\end{proof}

\section{Splitting measurements}
\label{sec:splitting_measurements}

The category $\FHilbCPTN$ does not have coproducts, but it can nevertheless be instructive to see why a given candidate for cotupling fails. Working in $\Lt(\Contraction)$, given two Stinespring dilations $f \colon A \to C \otimes G$ and $g \colon B \to C \otimes G'$, a candidate for the Stinespring dilation of their cotupling is the map $\delta_L^{-1} \circ f \oplus g \colon A \oplus B \to C \otimes (G \oplus G')$. In particular, if we were to try to cotuple the trivial Stinespring dilations of the canonical injections $i_1 \colon A \to (A \oplus B) \otimes I$ and $i_2 \colon B \to (A \oplus B) \otimes I$ in $\Contraction$, the result would be the map
$\delta_L^{-1} \circ i_1 \oplus i_2 : A \oplus B \to (A \oplus B) 
\otimes (I \oplus I) \enspace.$
Notice the non-trivial nature of the ancilla of this dilation, which tracks whether the result came from $i_1$ or $i_2$. This is not simply the identity, as it should be if this were an actual cotupling. Computing, we see that this map acts on block diagonal density matrices on $A \oplus B$ by measuring whether the result falls in $A$ or in $B$, \textit{i.e.}, by the mapping
$$\left(
\begin{array}{c|c}
  X & Y \\ \hline Z & W
\end{array}
\right) \mapsto \left(
\begin{array}{c|c}
  X & 0 \\ \hline 0 & W
\end{array}
\right) \text.$$
Clearly, maps $e$ such as these are both idempotent and \emph{causal} (in that they satisfy $\discard \circ e = \discard$), but interestingly they do not split in $\FHilbCPTN$. However, they \emph{do} have a very natural splitting in the category $\FCstarCPTN$ of finite-dimensional C*-algebras and CPTN maps, via $m : \mathcal{B}(A \oplus B) \to \mathcal{B}(A) \oplus \mathcal{B}(B)$ and $p : \mathcal{B}(A) \oplus \mathcal{B}(B) \to \mathcal{B}(A \oplus B)$ given by
\begin{equation}\label{eq-mp}
m\left(
\begin{array}{c|c}
  X & Y \\ \hline Z & W
\end{array}
\right) = (X, W) \qquad\text{and}\qquad
p(X,W) = \left(
\begin{array}{c|c}
  X & 0 \\ \hline 0 & W
\end{array}
\right)
\end{equation}
as we then have $m \circ p = \id$ and $p \circ m = e$. Indeed, as argued in \cite{coeckeselbytull:tworoads}, this is a defining characteristic of finite-dimensional C*-algebras compared to Hilbert spaces.

\begin{definition}[\cite{coeckeselbytull:tworoads}]
  The idempotent splitting of causal idempotents in $\cat{C}$ (where $\cat{C}$
  is a category with discarding) is the category
  $\SplitM(\cat{C})$ whose
  \begin{itemize}
    \item objects are pairs $(A,e)$ of an object $A$ of $\cat{C}$ and an 
    idempotent $e \colon A \to A$ satisfying $\discard \circ e = \discard$,
    \item morphisms $f \colon (A,e) \to (B,e')$ are morphisms $f \colon A \to
    B$ of $\cat{C}$ satisfying $e' \circ f \circ e = f$,
    \item each identity $(A,e) \to (A,e)$ is $e$, and
    \item composition is as in $\cat{C}$.
  \end{itemize}
\end{definition}

There is an inclusion of $\cat{C}$ in $\SplitM(\cat{C})$ sending objects $A$ to $(A, \id)$ and leaving morphisms unchanged. This construction is well-known to be the free splitting of these idempotents; more abstractly, it is the completion of a category with regards to certain absolute colimits~\cite{pare:abs}. That this constructs finite dimensional C*-algebras out of Hilbert spaces follows by \cite{coeckeselbytull:tworoads}:

\begin{proposition}
  There is a (monoidal, causal) equivalence
  $\SplitM(\FHilbCPTN) \simeq \FCstarCPTN$.
\end{proposition}
\begin{proof}
  It suffices to argue that the objects of $\SplitM(\FHilbCPTN)$ are
  finite-dimensional C*-algebras, which is immediate by Corollary 3.4 of
  \cite{coeckeselbytull:tworoads}.
\end{proof}

\section{Discussion}
\label{sec:discussion}
We have presented a universal construction that constructs the category of finite-dimensional C*-algebras and completely positive trace-nonincreasing maps from the category of finite-dimensional Hilbert spaces and unitaries. Though we have kept to the finite-dimensional case in this paper, there is reason to suspect that many of these results will generalise to infinite dimensions. For example, \emph{Halmos dilation} has a far stronger statement as \emph{Sz.\ Nagy dilation} in the infinite-dimensional case, and the usual Stinespring dilation theorem generalises to the infinite-dimensional case as well.

The key application that we envision for this work is in the design and semantics of quantum programming languages. One application of the $\LRa$-construction is in the quantisation of reversible classical programs. It can be shown that applying the $\LRa$-construction to the category $\FinBij$ of finite sets and bijections yields the category $\FinPInj$ of finite sets and partial injective functions. In particular, this means that the quantisation functor $\FinBij \to \Unitary$ lifts uniquely to a functor $\FinPInj \to \Contraction$. This is interesting since $\FinPInj$ is \emph{the} setting for (finite) reversible classical computing (see \textit{e.g.}~\cite{kaarsgaardaxelsenglueck:jinv,glueckkaarsgaardyokoyama:revsem,glueckkaarsgaardyokoyama:revmeta}), in particular \emph{reversible (classical) flowcharts}~\cite{yokoyamaaxelsenglueck:fundamentals,glueckkaarsgaard:categorical} (with a finite state space). This suggests that $\Contraction$ may be similarly considered as a setting for \emph{reversible quantum flowcharts}, a kind of quantum flowcharts (see also \cite{selinger:qpl}) which eschew measurements in favour of quantum control.

Another application concerns extending the quantum programming language $\mathcal{U}\Pi^\chi_a$ (``yuppie-chi-a'')~\cite{heunenkaarsgaard:qie} with \emph{classical types}. $\mathcal{U}\Pi^\chi_a$ is an effectful extension of $\mathcal{U}\Pi$, a quantum extension to the strongly typed classical reversible programming language $\Pi$~\cite{jamessabry:infeff}. The effectful part of $\mathcal{U}\Pi^\chi_a$ is that it uses the $\Lm$ and $\Ra$ constructions in a very direct way to add measurement as an effect to an otherwise measurement-free language, taking its semantics from $\Unitary$ to $\FHilbCPTP$. Since the causal idempotents described in Section~\ref{sec:splitting_measurements} can all be described as $\mathcal{U}\Pi^\chi_a$ programs, the $\SplitM$-construction can similarly be used very directly to add classical types to this language as a computational effect. This addresses a known shortcoming of the language by allowing a distinction between \textit{e.g.}\ the type of \emph{bits} and the type of \emph{qubits} (as in \textit{e.g.}~\cite{selinger:qpl}), but also allows for a much more fine-grained type-level separation of measurement maps in \emph{how much} they measure.

\paragraph{Acknowledgements}
We are indebted to the anonymous reviewers for their comments and suggestions, to John van der Wetering for a lively discussion at the event, and to Sean Tull for discussions and pointing out the equivalence between splitting measurements and splitting causal idempotents in $\FHilbCPTN{}$.

\bibliographystyle{eptcs}
\bibliography{refs}

\appendix

\section{Deferred proofs}
\begin{repproposition}{prop-lra-zero}
  The additive unit $O$ is a zero object in $\LRa(\cat{C})$.
\end{repproposition}
\begin{proof}
  Define maps $A \to O$ and $O \to A$ as the equivalence classes 
  of the symmetry $\sigma_\oplus \colon A \oplus O \to O \oplus A$ and 
  $\sigma_\oplus \colon O \oplus A \to A \oplus O$. To see that the map $A \to O$ is 
  unique, let $f \colon A \oplus H \to O \oplus G$ be any other morphism of
  $\cat{C}$. Then $[O, \sigma, A] \medeq{LR} [H, f, G]$ can be seen using the graphical language~\cite{heunenvicary:cqt}:
  \ctikzfig{zero}
  That the map $O \to A$ is unique follows analogously, so $O$ is both initial and terminal, \textit{i.e.}, a zero object.
\end{proof}

\begin{repproposition}{prop-lra-dagger}
  When $\cat{C}$ is a dagger rig category, so is $\LRa(\cat{C})$.
\end{repproposition}
\begin{proof}
  For any morphism \([H,f,G] \colon A \to B\) in \(\LRa(\C)\), its dagger is 
  defined to be $$[H,f,G]^\dagger = [G,f^\dagger,H].$$
  First, we need to check this is well-defined, \textit{i.e.}, given \([H,f,G] = [H',f',G']\), verify \([H,f,G]^\dagger = [H',f',G']^\dagger\). This is straightforward, as \(f \medeq{LR} g\) implies \(f^\dagger \medeq{LR} g^\dagger\); to see this, notice that
  \begin{align*}
    f \medeq{L} g &\implies f^\dagger \medeq{R} g^\dagger \\
    f \medeq{R} g &\implies f^\dagger \medeq{L} g^\dagger
  \end{align*}
  by using the dagger of the mediators and, furthermore, \(f \medeq{\id} g\) 
  trivially implies \(f^\dagger \medeq{\id} g^\dagger\).
  Clearly, this dagger is involutive and
  $$
    \id_A^\dagger = [O,\id_{A\oplus O},O]^\dagger = [O,\id_{A\oplus 
    O}^\dagger,O] = \id_A.
  $$
  Checking explicitly that \((g \circ f)^\dagger = f^\dagger \circ g^\dagger\) is more involved, but conceptually trivial: flip the diagram and use the fact that \(\C\) is a dagger category. This works because, in \(\LRa(\C)\) there is symmetry between input and output: the dagger turns hidden parts of the input into a hidden part of the output, and vice versa. That coherence isomorphisms are unitary follows immediately by the fact that they are inherited from $\cat{C}$, where are they unitary by $\cat{C}$ a dagger rig category.
\end{proof}

\begin{repproposition}{prop-LR-eq}
  When $\C$ is a rig category, \(\La(\Ra(\C)) \cong \LRa(\C) \cong \Ra(\La(\C))\).
\end{repproposition}
\begin{proof}
  The fact that \(O\) is initial in \(\LRa(\C)\) together with the universal property of \(\Ra\) implies that the top triangle of the diagram below commutes, for a unique functor \(\Phi\). Then, the fact that \(O\) is also terminal in \(\LRa(\C)\) and the universal property of \(\La\) implies that the bottom triangle also commutes, for a unique functor \(\widehat{\cF}\).
  \[\begin{tikzcd}
    \C && {\Ra(\C)} \\
    \\
    {\LRa(\C)} && {\La(\Ra(\C))}
    \arrow["\cD", from=1-1, to=1-3]
    \arrow["\cF"', from=1-1, to=3-1]
    \arrow["\Phi"', dashed, from=1-3, to=3-1]
    \arrow["\cE", from=1-3, to=3-3]
    \arrow["{\widehat{\cF}}", dashed, from=3-3, to=3-1]
  \end{tikzcd}\]
Moreover, \(O\) is a zero object in \(\La(\Ra(\C))\) and, thus, the universal property of \(\LRa\) implies there is a unique strong monoidal functor \(\widehat{\cE \circ \cD} \colon \LRa(C) \to \La(\Ra(\C))\) such that \(\cE \circ \cD = \widehat{\cE \circ \cD} \circ \cF\). This, together with the commuting diagram above implies:
$$\widehat{\cE \circ \cD} \circ \widehat{\cF} \circ \cE \circ \cD = \widehat{\cE \circ \cD} \circ \Phi \circ \cD = \widehat{\cE \circ \cD} \circ \cF = \cE \circ \cD
$$
which is captured in the outer triangle of the diagram below
  \[\begin{tikzcd}
    \C && {\Ra(\C)} && {\La(\Ra(\C))} \\
    \\
    &&&& {\La(\Ra(\C))}
    \arrow["\cD", from=1-1, to=1-3]
    \arrow["\cE", from=1-3, to=1-5]
    \arrow["{\widehat{\cE \circ \cD} \circ \widehat{\cF}}", dashed, from=1-5, to=3-5]
    \arrow["{\cE \circ \cD}"', from=1-1, to=3-5]
    \arrow[dashed, from=1-3, to=3-5]
  \end{tikzcd}\]
Due to the universal property of \(\Ra\), the functor shown above as a diagonal dashed line is unique, so it must be \(\cE\); then, \(\widehat{\cE\circ\cD} \circ \widehat{\cF} = 1_{\La(\Ra(\C))}\) as both functors make the inner right triangle commute but, by the universal property of \(\La\), there is only one such functor. By a similar argument \(\widehat{\cF} \circ \widehat{\cE \circ \cD} = 1_{\LRa(\C)}\) and hence \(\La(\Ra(\C)) \cong \LRa(\C)\). This is an equivalence of rig categories, as \(\widehat{\cF}\) and \(\widehat{\cE \circ \cD}\) are rig functors by construction.
The same strategy proves \(\Ra(\La(\C)) \cong \LRa(\C)\).
\end{proof}

\begin{reptheorem}{thm-pushout}
  $\Lt(\cat{C})$ is a pushout of $\mathcal{E} \colon \DagMon(\cat{C}) \to
  \Lm(\DagMon(\cat{C}))$ along the inclusion functor $\DagMon(\cat{C})
  \rightarrowtail \cat{C}$ in the category of locally small symmetric monoidal 
  categories and strong monoidal functors.
  \[\begin{tikzcd}[ampersand replacement=\&]
  	{\DagMon(\mathbf{C})} \& {\mathbf{C}} \\
  	{\Lm(\DagMon(\mathbf{C}))} \& {\Lt(\mathbf{C})} \\
  	\&\& {\mathbf{D}}
  	\arrow["I", tail, from=1-1, to=1-2]
  	\arrow["{\mathcal{E}}"', from=1-1, to=2-1]
  	\arrow["{I_t}"', tail, from=2-1, to=2-2]
  	\arrow["{\mathcal{E}_t}", from=1-2, to=2-2]
  	\arrow["{F}"', curve={height=12pt}, from=2-1, to=3-3]
  	\arrow["F_t", curve={height=-12pt}, from=1-2, to=3-3]
  	\arrow["{\hat{F}}", dotted, from=2-2, to=3-3]
  \end{tikzcd}\]
\end{reptheorem}
\begin{proof}
  Notice first that the upper square commutes since all functors involved
  are strict monoidal, and $I$ and $I_t$ are merely inclusions behaving as the
  identity on objects and morphisms, while $\mathcal{E}$ is precisely
  $\mathcal{E}_t$ restricted to dagger monomorphisms of $\cat{C}$.
  
  Next, since objects on $\cat{C}$, $\DagMon(\cat{C})$, $\Lt(\cat{C})$, and
  $\Lm(\DagMon(\cat{C}))$ all coincide, $F$ and $F_t$ must agree on objects, so 
  we may
  define $\hat{F}(X) = F(X) = F_t(X)$ on objects, and $\hat{F} \circ I_t = F$
  and $\hat{F} \circ \mathcal{E}_t = F_t$ on objects follows immediately. On
  morphisms we define
  $\hat{F}([f,E]) = F(\pi_1) \circ F_t(f)$. Then
  \begin{align*}
    \hat{F}(I_t([f,E])) & = \hat{F}(I_t(\pi_1 \circ \mathcal{E}(f)))
     = \hat{F}(I_t(\pi_1) \circ I_t(\mathcal{E}(f)))
     = \hat{F}(I_t(\pi_1) \circ \mathcal{E}_t(I(f))) \\
     & = \hat{F}([I(f),E]) = F(\pi_1) \circ F_t(I(f))
     = F(\pi_1) \circ F(\mathcal{E}(f)) \\
     & = F(\pi_1 \circ \mathcal{E}(f)) = F([f,E])
  \end{align*}
  so $\hat{F} \circ I_t = F$ on morphisms as well. For the other triangle,
  \begin{align*}
    \hat{F}(\mathcal{E}_t(f)) & = \hat{F}(\pi_1 \circ \rho^{-1} \circ 
    \mathcal{E}_t(f)) = \hat{F}(\pi_1 \circ \mathcal{E}_t(\rho^{-1}) \circ 
    \mathcal{E}_t(f)) = \hat{F}(\pi_1 \circ \mathcal{E}_t(\rho^{-1} \circ f)) \\
    & = \hat{F}([\rho^{-1} \circ f,I]) = F(\pi_1) \circ F_t(\rho^{-1} \circ f)
    = F(\pi_1) \circ F(\rho^{-1}) \circ F_t(f) \\
    & = F(\pi_1) \circ F_t(I(\rho^{-1})) \circ F_t(f)
    = F(\pi_1) \circ F(\mathcal{E}(\rho^{-1})) \circ F_t(f) \\
    & = F(\pi_1 \circ \mathcal{E}(\rho^{-1})) \circ F_t(f)
    = F([\rho^{-1},I]) \circ F_t(f) = F(\id) \circ F_t(f) = F_t(f)
  \end{align*}
  establishing $\hat{F} \circ \mathcal{E}_t = F_t$ on morphisms. Finally, suppose $G$ satisfies $G \circ I_t = F$ and $G \circ
  \mathcal{E}_t = F_t$. Then
  \begin{align*}
    G([f,E]) = G(I_t(\pi_1) \circ \mathcal{E}_t(f)) = 
    G(I_t(\pi_1)) \circ G(\mathcal{E}_t(f)) = F(\pi_1) \circ F_t(f) = 
    \hat{F}([f,E])
  \end{align*}
  on morphisms, and since $G$ and $\hat{F}$ must agree on objects as well 
  (\textit{e.g.}\ $\hat{F}(X) = F(X) = G(I_t(X)) = G(X)$), we get $G = \hat{F}$. 
  Finally, $\hat{F}$ is strong monoidal since $\mathcal{E}_t$ is strict 
  monoidal and $F_t$ and $\hat{F}$ agree on objects, so $\hat{F}$ may reuse the
  coercions $I \cong F_t(I)$ and $F_t(A \otimes B) \cong F_t(A) \otimes F_t(B)$.
\end{proof}

\end{document}